\documentclass[runningheads]{llncs}

\usepackage[T1]{fontenc}
\usepackage[utf8]{inputenc}
\usepackage{amsmath,amssymb,mathtools}
\usepackage{graphicx}
\usepackage{booktabs}
\usepackage{xspace}
\usepackage{url}
\usepackage{tikz}
\usepackage{subcaption}

\newcommand{\Oh}{\ensuremath{\mathcal{O}}}
\newcommand{\whp}{w.h.p.\xspace}
\newcommand{\lce}{\ensuremath{\mathsf{lce}}}

\newcommand{\makestring}{\ensuremath{\mathsf{make\_string}}}
\newcommand{\splitop}{\ensuremath{\mathsf{split}}}
\newcommand{\concatop}{\ensuremath{\mathsf{concat}}}

\newcommand{\equalop}{\ensuremath{\mathsf{equal}}}
\newcommand{\access}{\ensuremath{\mathsf{access}}}

\newcommand{\fp}{\ensuremath{\mathsf{fp}}}
\newcommand{\pow}{\ensuremath{\mathsf{pow}}}
\newcommand{\size}{\ensuremath{\mathsf{size}}}
\newcommand{\height}{\ensuremath{\mathsf{height}}}

\newcommand{\val}{\ensuremath{\mathsf{val}}}

\newcommand\textal[1]{\ensuremath{\mathsf{#1}}}

\newcommand\ul[1]{\underline{#1}}

\usepackage{orcidlink}
\providecommand{\orcidID}[1]{}
\renewcommand{\orcidID}[1]{\orcidlink{#1}}

\begin{document}

\title{Fully Persistent Dynamic LCE \\ via AVL Trees and AVL Grammars}
\titlerunning{Fully Persistent Dynamic LCE}

\author{Taiki Kaneda\inst{1} \and Hiroki Arimura\inst{1}\orcidID{0000-0002-2701-0271} \and Shunsuke Inenaga\inst{2}\orcidID{0000-0002-1833-010X}}
\authorrunning{T. Kaneda et al.}

\institute{Graduate School of IST, Hokkaido University, Sapporo, Japan\\
\email{kaneda.taiki.c9@elms.hokudai.ac.jp, arim@ist.hokudai.ac.jp}
\and
Department of Informatics, Kyushu University, Fukuoka, Japan\\
\email{inenaga.shunsuke.380@m.kyushu-u.ac.jp}}

\maketitle

\begin{abstract}
We study fully persistent dynamic strings with equality and longest common extension (LCE) queries. Straightforward full persistence is problematic for the splay-based FeST structure, since the same 
unbalanced past version
can be reused indefinitely and the usual amortized analysis no longer applies. We give a fully persistent dynamic LCE structure, called \textal{FeAVL}, based on path copying over AVL trees. 
For an operation involving string(s) of total length $n$, it supports split, concatenate, and single-character updates in worst-case $\Oh(\log n)$ time, equality in worst-case $\Oh(\log n)$ time \whp, and LCE in worst-case $\Oh(\log n+\log^2\ell)$ time \whp, where $\ell$ is the answer; each update creates only $\Oh(\log n)$ new permanent nodes. We also give a grammar-compressed instantiation via AVL grammars: starting from an initial grammar of size $g_0$, after $U$ updates, the total number of permanent grammar nodes is $\Oh(g_0+I+U\log n_{\max})$, where $I$ is the number of inserted fresh characters and $n_{\max}$ is the maximum string length appearing during the update sequence.
\end{abstract}


\section{Introduction}

\textit{Dynamic strings} support insertion, deletion, substitution, splitting, and concatenation.
A fundamental query on such strings is the longest common extension (LCE) query: given two positions, possibly in two different strings, return the length of the longest common prefix of the two corresponding suffixes.
Dynamic LCE and dynamic string equality are basic primitives in dynamic text processing, versioned documents, genome editing, and dynamic indexing.

The dynamic string problem has a long history, from practical \textit{rope data structures}~\cite{boehm1995ropes} to the \textit{optimal structure} of Gawrychowski, Karczmarz, Kociumaka, {\L}{\k{a}}cki, and Sankowski~\cite{gawrychowski2018optimal}, which achieves logarithmic updates and constant-time equality and LCE queries \whp using locally consistent parsing~\cite{sahinalp1995data}.
In contrast, Lipt\'{a}k, Masillo, and Navarro proposed \textal{FeST} 
(a \textit{\ul{f}orest of \ul{e}nhanced \ul{s}play \ul{t}rees}) 
as a simple textbook solution~\cite{liptak2026textbook}.
\textal{FeST} supports splits, concatenations, equality, and LCE queries with logarithmic or nearly logarithmic amortized bounds.

This paper focuses on making this simple paradigm fully persistent.
A \textit{fully persistent data structure} allows updates on any version, producing a new branch of the version history.
Throughout this paper, $n$ denotes the total length of the string(s) involved in the operation under consideration.
For a sequence of $U$ updates, we write $n_{\max}=\max_{0 \leq i \leq U} n_u$, 
where $n_u$ is the corresponding value of $n$ for the $u$-th update.

Amortized structures like splay trees~\cite{sleator1985self} lose efficiency in full persistence because poorly balanced versions can be repeatedly queried.
Therefore, we propose \textal{FeAVL} (a \textit{\ul{f}orest of \ul{e}nhanced \ul{AVL} trees}), replacing self-adjusting trees with worst-case balanced implicit-key trees such as AVL trees~\cite{adelson1962avl} augmented with sizes and fingerprints.
Path copying gives full persistence because rank searches, splits, concatenations, and rebalancing affect only $\Oh(\log n)$ nodes


The choice of AVL trees also has a compressed counterpart: AVL grammars, introduced by Rytter for grammar-based compression~\cite{rytter2003application}.
In the resulting fully persistent data structure, called the \textit{\textal{FeAVL} grammar}, an LZ77 factorization of size $z_0$ for an initial string of length $n_0$ can be transformed into an AVL grammar of size $\Oh(z_0\log n_0)$.
By choosing the smaller of this grammar and the uncompressed AVL tree, we may assume an initial size of $g_0=\Oh(\min\{n_0,z_0\log n_0\})$.
For subsequent updates, we only rely on the fact that split and concatenation operations function correctly for arbitrary AVL grammars.
Table~\ref{tab:intro-comparison} summarizes the two layers of the result,
and our contributions are summarized as follows:


\begin{table}[t!]
\centering
\caption{Comparison of the AVL-tree version and the AVL-grammar version. $I$~is the number of inserted fresh characters.}
\label{tab:intro-comparison}
\footnotesize
\resizebox{\textwidth}{!}{%
\begin{tabular}{lll}
\toprule
Measure & AVL-tree version & AVL-grammar version \\
\midrule
Initial size & $\Oh(n_0)$ & $\Oh(g_0)$ \\
Best initial choice & $\Oh(n_0)$ & $\Oh(\min\{n_0,z_0\log n_0\})$ \\
Split/Concatenate time & $\Oh(\log n)$ & $\Oh(\log n)$ \\
Update time & $\Oh(\log n)$ & $\Oh(\log n)$ \\
Equality time & $\Oh(\log n)$, \whp & $\Oh(\log n)$, \whp \\
LCE time & $\Oh(\log n+\log^2\ell)$, \whp & $\Oh(\log n+\log^2\ell)$, \whp \\
New permanent objects/update & $\Oh(\log n)$ nodes & $\Oh(\log n)$ symbols \\
Total size after $U$ updates & $\Oh(n_0+I+U\log n_{\max})$ & $\Oh(g_0+I+U\log n_{\max})$ \\
\bottomrule
\end{tabular}%
}
\end{table}

\begin{enumerate}
\item We give an AVL-tree version of \textal{FeST}, called \textal{FeAVL}, with worst-case rather than amortized tree bounds: equality in $\Oh(\log n)$ time \whp and LCE in $\Oh(\log n+\log^2 \ell)$ time \whp, where $\ell$ is the answer.
\item We make the structure fully persistent by path copying. Every update can be issued from any previous version and creates only $\Oh(\log n)$ new permanent nodes.
\item We isolate the obstruction to a straightforward fully persistent splay-based \textal{FeST}: path copying may copy a linear-height access path, and full persistence allows the same bad old version to be reused repeatedly.
\item We give a grammar-compressed instantiation using AVL grammars. Starting from an initial grammar of size $g_0$, after $U$ updates, the structure occupies $\Oh(g_0+I+U\log n_{\max})$ permanent grammar nodes.
\end{enumerate}

\section{Preliminaries}\label{sec:prelim}

\subsection{Strings and Dynamic LCE}

Let $\Sigma$ be an alphabet.
An element of $\Sigma^*$ is a string.
The length of a string $s$ is denoted by $|s|$.
For $1\le i\le j\le |s|$, $s[i..j]$ denotes the substring $s[i]\cdots s[j]$
that starts at position $i$ and ends at position $j$ in $s$.
Let $s[i..] = s[i..|s|]$ and $s[..j] = s[1..j]$.
For strings $s$ and $t$, $st$ denotes their concatenation, of length $|s|+|t|$.
Given a position $i$ in a string $s$ and a position $j$ in a string $t$,
the \emph{longest common extension} (\emph{LCE}) query is to answer
\[
  \lce(s,i,t,j)=\max\{\ell\ge 0 : s[i..i+\ell-1]=t[j..j+\ell-1]\},
\]
which is the longest common prefix (LCP) length of $s[i..]$ and $t[j..]$.

We maintain a collection of string objects.
The following operations are used throughout the paper.
\begin{itemize}
\item $\makestring(w)$: create a new string object storing the static string $w$.
\item $\splitop(s,i)$: split string $s$ into $s[1..i]$ and $s[i+1..]$.
\item $\concatop(s,t)$: concatenate two strings $s$ and $t$, producing $st$.
\item $\access(s,i)$: return character $s[i]$.
\item $\equalop(s,i,t,j,\ell)$: test whether $s[i..i+\ell-1]=t[j..j+\ell-1]$.
\item $\lce(s,i,t,j)$: return the LCE length of strings $s[i..]$ and $t[j..]$.
\end{itemize}
Single-character insertion, deletion, substitution, extraction, and introduction of substrings are obtained from these primitives in the standard way.

Unless stated otherwise, for an operation on one or two strings we write $n$ for the total length of the string or strings involved in that operation.
Thus, for $\lce(s,i,t,j)$, we use $n=|s|+|t|$ if $s$ and $t$ are distinct string objects, and $n=|s|$ if the two positions belong to the same string.
For space bounds over a sequence of updates, $n_u$ denotes the corresponding value of $n$ for the $u$-th update, and $n_{\max}$ denotes the maximum string length involved over the whole history.
Throughout the paper, logarithmic terms in asymptotic bounds are understood to be at least one; equivalently, $\log x$ means $\max\{1,\log_2 x\}$ inside $\Oh(\cdot)$-bounds.

\subsection{Karp--Rabin Fingerprints}

We use Karp--Rabin fingerprints~\cite{karp1987rabin}.
Let $p$ be a prime and let $b\in\{1,\ldots,p-1\}$ be chosen uniformly at random.
For a string $x=x[1]\cdots x[m]$, define
\[
  \kappa(x)=\sum_{r=0}^{m-1} x[m-r] b^r \bmod p .
\]
For two strings $x$ and $y$, let $\kappa(xy)=(\kappa(x)b^{|y|}+\kappa(y))\bmod p$.
We store with each represented substring both its fingerprint and the value $b^{|x|}\bmod p$.
As usual, by choosing $p$ sufficiently large, all equality answers are correct \whp over any polynomial-length sequence of operations.
The data structure itself is deterministic except for this fingerprinting choice.

\subsection{AVL Trees as Implicit-Key Sequence Trees}

An AVL tree~\cite{adelson1962avl} is a binary tree in which, for every internal node, the heights of the two children differ by at most one (see~\cite{adelson1962avl,mehlhorn2008algorithms}).
We use AVL trees as implicit-key sequence trees: the in-order traversal spells the represented string, and ranks are in-order positions computed from subtree sizes~\cite{mehlhorn2008algorithms}.
The standard rank search, split, concatenate, and rebalancing operations visit and modify only $\Oh(\log n)$ nodes in the worst case.

We shall use the following property of AVL trees.
\begin{definition}[Logarithmic update-path property]\label{def:logpath}
A sequence tree family has the logarithmic update-path property if every tree containing $m$ nodes has height $\Oh(\log m)$, and rank-based search, split, concatenate, and rebalancing modify only $\Oh(\log m)$ nodes and run in $\Oh(\log m)$ worst-case time.
\end{definition}
AVL trees satisfy Definition~\ref{def:logpath}.
The same uncompressed construction extends to other joinable balanced trees satisfying the logarithmic update-path property.

\subsection{AVL Grammars}

A straight-line program, or grammar, is a directed acyclic graph of rules in which each rule is either a terminal rule $X\to a$ or a binary rule $X\to YZ$.
Each symbol $X$ derives a unique string, denoted $\val(X)$.
The size of the grammar is the number of symbols.
For a terminal rule we set $h(X)=0$, and for a binary rule $X\to YZ$ we set $h(X)=1+\max\{h(Y),h(Z)\}$.

\begin{definition}[AVL grammar]
A grammar is an AVL grammar if $|h(Y)-h(Z)|\le 1$ for every binary rule $X\to YZ$.
\end{definition}

An AVL grammar is therefore a DAG representation of a string whose unfolded derivation tree is an AVL tree.
The graph itself need not be a tree because the same symbol may be shared by many parents.
Rytter introduced AVL grammars for grammar-based compression and proved the following result.

\begin{theorem}[Rytter~\cite{rytter2003application}]\label{thm:rytter}
Given an LZ77 factorization of size $z_0$ for a string $s$ of length $n_0$, one can construct an AVL grammar of size $\Oh(z_0\log n_0)$ deriving the string $s$.
\end{theorem}

In this paper, Theorem~\ref{thm:rytter} is used only to obtain a small initial representation.
All subsequent dynamic and persistent operations require only that the current representation is an AVL grammar.

\subsection{Persistence}\label{sec:persistence}

A partially persistent structure allows queries on old versions, but updates only on the newest version.
A fully persistent structure allows both queries and updates on any old version.
The version graph is therefore a tree; each update derives from a single existing version.
We use the standard path-copying method~\cite{driscoll1989persistent}: when an update changes a node, a fresh copy of that node is created, and the change is propagated through copied ancestors up to the root.
Old nodes are never overwritten.

\section{An AVL-Tree Data Structure for Dynamic LCE}\label{sec:dynamic}

In this section, we propose our fully persistent data structure, called \textal{FeAVL} 
(a \textit{\ul{f}orest of \ul{e}nhanced \ul{AVL} trees}).  
In our structure, 
each string is stored in one AVL tree.
The in-order traversal of the tree spells the string.
For a node $v$, let $L(v)$ and $R(v)$ be its left and right children.
The node $v$ stores five values $v.\mathsf{char}$, $v.\size$, $v.\height$, $v.\fp$, and $v.\pow$,
where $v.\size$ is the number of nodes in the subtree rooted at $v$, $v.\height$ is the AVL height, $v.\fp$ is the fingerprint of the substring represented by the subtree rooted at $v$, and $v.\pow=b^{v.\size}\bmod p$,
where $b$ is the Karp--Rabin base fixed in Sec.~\ref{sec:prelim}.
Null children have size 0, fingerprint 0, and power 1.
Each time $v$ obtains a new child, its fields are updated in constant time:
\begin{align*}
  v.\size &= \size(L(v)) + \size(R(v)) + 1,\\
  v.\fp &= ((\fp(L(v))\cdot b + v.\mathsf{char})\cdot \pow(R(v)) + \fp(R(v))) \bmod p,\\
  v.\pow &= \pow(L(v))\cdot b\cdot \pow(R(v)) \bmod p.
\end{align*}
The height field is updated in the standard AVL way.
All rotations and joins call this constant-time update routine on the nodes whose children changed.


\subsection{Basic Operations}

\paragraph{Construction.}
$\makestring(w)$ builds a balanced tree whose in-order labels are the characters of $w$.
A perfectly balanced tree can be constructed in $\Oh(|w|)$ time, and all auxiliary fields are computed by a post-order traversal.

\paragraph{Access.}
$\access(s,i)$ follows the usual rank search using subtree sizes.
Because the tree has logarithmic height, the time is $\Oh(\log |s|)$.

\paragraph{Split and concatenate.}
$\splitop(s,i)$ applies rank-based split to the tree representing $s$ and returns two roots.
$\concatop(s,t)$ applies AVL concatenation to the two roots.
Both operations preserve all auxiliary fields by recomputing them on the modified nodes.
By Definition~\ref{def:logpath}, each operation takes $\Oh(\log |st|)$ worst-case time.

\begin{lemma}\label{lem:splitconcat}
In the enhanced AVL tree, rank split and concatenate take $\Oh(\log n)$ worst-case time and modify $\Oh(\log n)$ nodes, where $n$ is the total number of nodes in the input tree or trees.
\end{lemma}

\begin{proof}
The underlying AVL operation performs rank search, split, concatenate, and rebalancing in $\Oh(\log n)$ worst-case time and modifies only nodes on logarithmic search and rebalancing paths.
The added information is recomputed in constant time for each modified node.
\qed
\end{proof}

\subsection{Substring Fingerprints and Equality}

To compute the fingerprint of substring $s[i..j]$, we temporarily split the tree into three parts
$s[..i-1]$, $s[i..j]$, and $s[j+1..]$.
The middle root stores the fingerprint of $s[i..j]$.
In the non-persistent structure the three parts can be concatenated back.

\begin{lemma}\label{lem:rangefp}
The fingerprint of any substring of a represented string $s$ can be obtained in $\Oh(\log |s|)$ worst-case time.
\end{lemma}

\begin{proof}
Two rank splits isolate the substring, and the fingerprint is stored at the root of the isolated tree.
The cost follows from Lemma~\ref{lem:splitconcat}.
\qed
\end{proof}

\begin{corollary}\label{cor:equal}
By comparing the range fingerprints of $s[i..i+\ell-1]$ and $t[j..j+\ell-1]$ from Lemma~\ref{lem:rangefp}, $\equalop(s,i,t,j,\ell)$ is answered in $\Oh(\log n)$ worst-case time and is correct \whp
\end{corollary}

\subsection{LCE Queries}\label{subsec:lce}

The LCE query follows the LCP algorithm of Lipt\'{a}k
et al.~\cite{liptak2026textbook}.
Their algorithm uses substring equality tests together with temporary
extraction and reinsertion of windows.
It first obtains a crude upper bound on the answer by testing candidate
lengths that grow by repeated squaring, and then finishes the search inside
localized windows.
To avoid an extra factor depending on the global length when the answer is
small, it applies this localization idea twice, the ``hitting twice'' strategy
of \textal{\textal{FeST}}.

The same algorithm applies to our enhanced AVL trees.
Indeed, the primitives used in the \textal{\textal{FeST}} LCP algorithm---substring equality,
split, concatenate, extract, and introduce---are supported by Lemmas~\ref{lem:splitconcat}
and~\ref{lem:rangefp} in worst-case logarithmic time.
Thus the amortized analysis of Lipt\'{a}k et al.\ becomes a worst-case analysis
in our setting.

\begin{lemma}\label{lem:lce}
An LCE query with answer $\ell$ can be answered in
$\Oh(\log n+\log^2\ell)$ worst-case time and is correct \whp
\end{lemma}

\begin{proof}
Apply the LCP algorithm of Lipt\'{a}k et al.\ to the two isolated suffixes.
Their analysis gives $\Oh(\log n+\log^2\ell)$ provided substring equality and the temporary window operations run in time logarithmic in the relevant tree size.
In our setting these are the range fingerprints of Lemma~\ref{lem:rangefp} and the split and concatenate of Lemma~\ref{lem:splitconcat}, which meet this bound in the worst case rather than amortized, so the same bound holds in the worst case.
The only possible error is a Karp--Rabin collision.
\qed
\end{proof}

\begin{theorem}\label{thm:dynamic}
There is a dynamic string data structure using space linear in the total length of the represented strings that supports $\makestring(w)$ in $\Oh(|w|)$ worst-case time, split and concatenate in $\Oh(\log n)$ worst-case time, substring equality in $\Oh(\log n)$ worst-case time \whp, and LCE in $\Oh(\log n+\log^2\ell)$ worst-case time \whp, where $\ell$ is the LCE answer.
\end{theorem}

\section{Full Persistence by Path Copying}\label{sec:persistent}

We now make the structure fully persistent in the sense of Sec.~\ref{sec:persistence}.
A version is a set of root pointers, one per string object; an operation acts on strings of a given version and returns a new version, modifying no node reachable from an old version.

\subsection{Path-Copied AVL Trees}

We path-copy every primitive: a node whose child pointer or auxiliary field would change is copied, and untouched subtrees are shared with the old version.

\begin{lemma}\label{lem:pathcopy}
A path-copied split or concatenate on an enhanced AVL tree runs in $\Oh(\log n)$ worst-case time and creates $\Oh(\log n)$ new nodes, where $n$ is the number of nodes in the involved tree or trees.
\end{lemma}

\begin{proof}
By Lemma~\ref{lem:splitconcat} the non-persistent operation modifies $\Oh(\log n)$ nodes; path copying creates one fresh node per modified node, recomputing its fields in $\Oh(1)$ and reusing the rest.
\qed
\end{proof}

\subsection{Persistent Updates and Queries}

Updates such as insertion, deletion, substitution, split, concatenate, extract, and introduce are implemented as in the non-persistent structure with path-copied primitives, so each creates a new version while preserving the input.

Equality and LCE queries are non-destructive: a read-only fingerprint traversal adds no permanent nodes, and any temporary path-copied splits used instead are discarded after the query.

\begin{theorem}\label{thm:persistent}
There is a fully persistent dynamic LCE data structure with the same query and update time bounds as Theorem~\ref{thm:dynamic}.
Each update creates $\Oh(\log n)$ new permanent nodes, where $n$ is the total length of the strings involved in the update.
If queries are implemented by temporary splits, an equality query creates $\Oh(\log n)$ temporary nodes and an LCE query creates $\Oh(\log n+\log^2\ell)$ temporary nodes; these nodes need not be retained as part of the persistent structure.
\end{theorem}

\begin{proof}
The algorithms are those of Sec.~\ref{sec:dynamic} with path-copied primitives, adding $\Oh(1)$ time per modified node; the permanent space bound follows from Lemma~\ref{lem:pathcopy}.
Correctness is by induction on the version graph: every new root is built from copied nodes and old immutable subtrees, and old roots are never touched.
\qed
\end{proof}

\paragraph{Space accounting.}
Let $U$ be the number of updates, let $n_0$ be the total length of initially created strings, and let $I$ be the total number of explicitly inserted fresh characters over all versions.
Let $n_u$ be the total length of the strings involved in update $u$.
The total number of permanent nodes is $\Oh(n_0+I+U \log n_{\max})$.
This form is preferable to a bound stated only in terms of the current total length, because fully persistent structures retain nodes belonging to old branches.

\subsection{Why Splay-Based \textal{FeST} Is Problematic}\label{sec:splay}

\textal{FeST} uses splay trees~\cite{sleator1985splay}, which give amortized logarithmic bounds in the ephemeral setting. This amortization is not robust under a straightforward path-copying implementation.

\begin{proposition}\label{prop:splaybad}
For every $n$, there is a version of a splay tree storing $n$ characters and a position such that a path-copying access to that position takes $\Omega(n)$ time and creates $\Omega(n)$ new nodes.
Moreover, under full persistence, the same version can serve as the input to $q$ independent accesses, giving $\Omega(qn)$ total time and space.
\end{proposition}


\begin{proof}
Accessing the positions in increasing rank order brings each accessed node to the root, leaving a left path with the smallest rank at the bottom; in the ephemeral structure this access sequence runs in $\Oh(n)$ total time~\cite{tarjan1985sequential}, so the path is an ordinary configuration. Take it as the target version.
An access to the smallest rank visits $\Theta(n)$ nodes and splays along the same path, so path copying creates $\Omega(n)$ new nodes in $\Omega(n)$ time.
Under full persistence the same version serves as the input to $q$ independent accesses, giving $\Omega(qn)$.
\qed
\end{proof}

The proposition rules out only the straightforward path-copied \textal{FeST}: for partial persistence, where updates follow a single evolving tree, the standard amortized analysis still applies.
Under full persistence, a worst-case balanced tree is the natural replacement.

\section{AVL-Grammar-Compressed Version}\label{sec:grammar}

We present a compressed version of a fully persistent data structure, 
the \textal{FeAVL} grammar, 
based on 
grammar-based compression~\cite{rytter2003application}.
It replaces the AVL tree of Sec.~\ref{sec:dynamic} by an AVL grammar DAG.
Each nonterminal $X$ stores the fields $|\val(X)|$, $h(X)$, $\fp(X)$, and $\pow(X)$, 
which 
are initialized directly for a terminal rule $X\to a$, and updated from the two children in 
$O(1)$ time
for a binary rule $X\to YZ$. 
\begin{align*}
  |\val(X)| &= |\val(Y)|+|\val(Z)|,\\
  h(X) &= 1+\max\{h(Y),h(Z)\},\\
  \fp(X) &= (\fp(Y)\cdot \pow(Z)+\fp(Z)) \bmod p,\\
  \pow(X) &= \pow(Y)\cdot \pow(Z) \bmod p.
\end{align*}
The rule is valid only when $|h(Y)-h(Z)|\le 1$.
Symbols are immutable, and a version stores root symbols for the strings in that version.

\subsection{Closure under Split and Concatenation}

The basic operation is Rytter's AVL-grammar concatenation.
Given two AVL grammars rooted at $X$ and $Y$, it creates an AVL grammar for $\val(X)\val(Y)$.
If the two heights differ by at most one, it creates the single rule $Z\to XY$.
If one side is much taller, the algorithm descends along the appropriate spine of the taller derivation tree, inserts the shorter grammar at the first height-compatible position, and then rebuilds and rebalances on the way back.

\begin{lemma}\label{lem:grammarconcat}
Let $X$ and $Y$ be roots of AVL grammars deriving strings of total length $n$.
One can create an AVL grammar root for $\val(X)\val(Y)$ in $\Oh(\log n)$ time using $\Oh(\log n)$ new symbols.
The input grammars are not modified.
\end{lemma}

\begin{proof}
This is the standard AVL-grammar concatenation procedure of Rytter~\cite{rytter2003application}, which is the grammar analogue of AVL-tree concatenation.
The height of every AVL derivation tree deriving a string of length at most $n$ is $\Oh(\log n)$, so the affected spine and the rebalancing path have logarithmic length.
All fields stored at newly created symbols are recomputed in constant time per symbol.
\qed
\end{proof}

We next use the interval-decomposition operation that already appears in Rytter's construction.
When processing a new LZ factor, Rytter locates its previous occurrence in the already constructed prefix grammar and decomposes the corresponding interval into nonterminals
$S_1, S_2,\ldots,S_t $
whose derived strings concatenate to the factor.
Because the grammar is AVL-balanced, this decomposition has $t=\Oh(\log n)$ fragments; moreover, the fragment heights have the right-sided/left-sided, equivalently bitonic, structure used in Rytter's analysis.
Joining this sequence telescopes: a concatenation of two grammars adds symbols proportional to their height difference, and along a bitonic sequence these differences sum to $\Oh(\log n)$.
Thus the fragments can be concatenated into one AVL grammar by adding only $\Oh(\log n)$ new symbols, not $\Oh(\log^2 n)$.
This argument is independent of the LZ origin of the interval: once an AVL grammar is given, the same top-down decomposition applies to any substring interval.

\begin{lemma}\label{lem:grammarsplit}
Let $X$ be the root of an AVL grammar deriving a string of length $n$.
For any $0\le i\le n$, one can create AVL grammar roots for $\val(X)[1..i]$ and $\val(X)[i+1..n]$ in $\Oh(\log n)$ time using $\Oh(\log n)$ new symbols.
The input grammar is not modified.
\end{lemma}

\begin{proof}
Apply Rytter's interval decomposition to the two intervals $[1,i]$ and $[i+1,n]$.
Each interval is represented as a sequence of $\Oh(\log n)$ existing grammar roots whose derived strings concatenate to that interval.
By the same right-sided/left-sided height analysis used by Rytter for concatenating the pieces of an LZ factor, each sequence can be turned into one AVL grammar root with $\Oh(\log n)$ new symbols and time.
The original grammar is immutable and is only referenced by the newly created symbols.
\qed
\end{proof}

Lemmas~\ref{lem:grammarconcat} and~\ref{lem:grammarsplit} hold for arbitrary AVL grammars; Rytter's construction is needed only for the initial size bound (Theorem~\ref{thm:rytter}).

\subsection{Persistent AVL Grammars}

We implement all grammar operations functionally, with every symbol immutable.
Split, concatenation, insertion, deletion, and substitution create new symbols and reuse all unaffected old symbols.
Therefore, every old version remains valid automatically, and any old root can be used as the input to a new update.

Substring equality and LCE follow Secs.~\ref{sec:dynamic} and~\ref{sec:persistent}, substituting grammar roots for tree roots.
Temporary grammar splits yield range fingerprints, and the LCE algorithm in Sec.~\ref{subsec:lce} uses the same 
strategy by Lipt\'{a}k et al.



\begin{theorem}\label{thm:grammarpersistent}
Suppose that the initial strings are represented by AVL grammars of total size $g_0$.
There is a fully persistent grammar-compressed dynamic LCE data structure supporting split, concatenate, and single-character updates in $\Oh(\log n)$ worst-case time, equality in $\Oh(\log n)$ worst-case time \whp, and LCE in $\Oh(\log n+\log^2\ell)$ worst-case time \whp
After updates $u=1,\ldots,U$, the total number of permanent grammar symbols is $\Oh(g_0+I+U\log n_{\max})$, where $I$ is the number of explicitly inserted fresh characters and $n_{\max}$ is the maximum string length appearing during the update sequence.
\end{theorem}

\begin{proof}
The correctness invariant is that every version root is the root of an AVL grammar deriving the corresponding string.
It holds for the initial roots by assumption.
It is preserved by split and concatenation by Lemmas~\ref{lem:grammarsplit} and~\ref{lem:grammarconcat}; insertions, deletions, and substitutions are compositions of these operations and a constant number of terminal creations.
The time bounds follow from the same lemmas and, for LCE, from the analysis of Lemma~\ref{lem:lce} applied to the grammar primitives.
The space bound counts the initial symbols, the explicitly inserted terminals, and $\Oh(\log n_u)$ new binary symbols per update.
\qed
\end{proof}

\begin{corollary}\label{cor:mininitial}
For an initial string of length $n_0$ whose LZ77 factorization has size $z_0$, the initial representation may be chosen to have size $\Oh(\min\{n_0,z_0\log n_0\})$.
Consequently, if $n_{\max}$ bounds the lengths of all strings involved in the updates, the permanent grammar size is $\Oh(\min\{n_0,z_0\log n_0\}+I+U\log n_{\max})$.
\end{corollary}

\begin{proof}
The uncompressed AVL tree is an AVL grammar of size $\Oh(n_0)$.
Rytter's construction gives an AVL grammar of size $\Oh(z_0\log n_0)$.
We choose the smaller of the two as the initial representation and then apply Theorem~\ref{thm:grammarpersistent}.
\qed
\end{proof}

We remark that 
this bound does not guarantee a size bounded by the LZ77 size $z_t$ of later versions.
Without rebuilding or dynamic parsing, the space depends solely on the initial grammar and the number of updates.


\section{Experiments}\label{sec:experiments}

We empirically evaluate \textal{FeAVL}, the AVL-tree structure of Secs.~\ref{sec:dynamic} and~\ref{sec:persistent}, measuring running time and peak memory against the string length $n$ and corroborating Theorems~\ref{thm:dynamic} and~\ref{thm:persistent} and Proposition~\ref{prop:splaybad}.
All plots vary one parameter on log--log axes, so the slope separates $\Oh(\log n)$ from $\Oh(n)$.
Each data point is the total wall-clock time of a fixed number $Q$ of queries, averaged over five runs on a single fixed-seed input.
Operands are uniform random, except for equality and the worst case, which use constructed inputs.

The experiments were run on Ubuntu 22.04.5 LTS (64-bit) with $64.0$\,GiB of memory and an AMD Ryzen Threadripper 3990X, compiled with \texttt{g++ 12.3.0}.
We implement three approaches in C++ and compare them: the naive array, where $\access$ costs $\Oh(1)$ and the other operations $\Oh(n)$ (with $\equalop$ in $\Oh(\ell)$); the splay-based FeST of~\cite{liptak2026textbook} and its path-copied persistent version; and our \textal{FeAVL}.

\subsection{Dynamic Operations}\label{ssec:dynamic}

With $Q = 10^4$ queries, extraction and introduction, a composition of split and concatenate, scale logarithmically for both trees and linearly for the naive array; the trees are faster beyond $n \sim 10^5$ (Figure~\ref{fig:dyn-eai}), consistent with Theorem~\ref{thm:dynamic}.
For equality we fix $n = 2\times10^6$ and vary the match length $\ell$ over a fully matching range.
\textal{FeAVL} and FeST stay nearly constant in $\ell$ (about $0.25$\,s and $12$\,ms, respectively; $\Oh(\log n)$ by fingerprints), while the naive array grows linearly to about $2.4$\,s at $\ell = 10^6$ ($\Oh(\ell)$), consistent with Corollary~\ref{cor:equal} (Figure~\ref{fig:dyn-equal}); persistent equality behaves identically.

\begin{figure}[t]
  \centering
  \begin{subfigure}{0.49\linewidth}
    \centering
    \includegraphics[width=\linewidth]{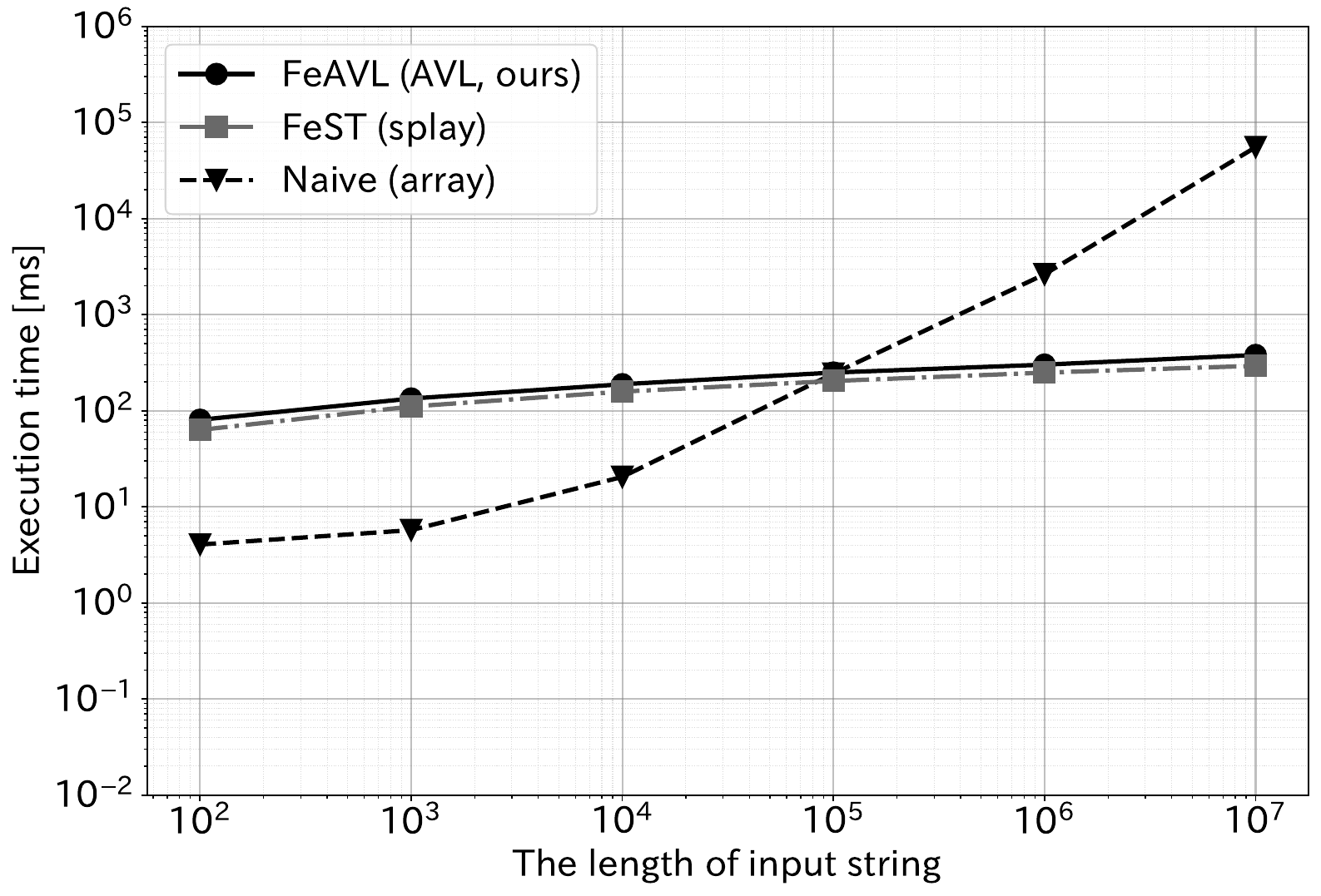}
    \caption{}\label{fig:dyn-eai}
  \end{subfigure}
  \hfill
  \begin{subfigure}{0.49\linewidth}
    \centering
    \includegraphics[width=\linewidth]{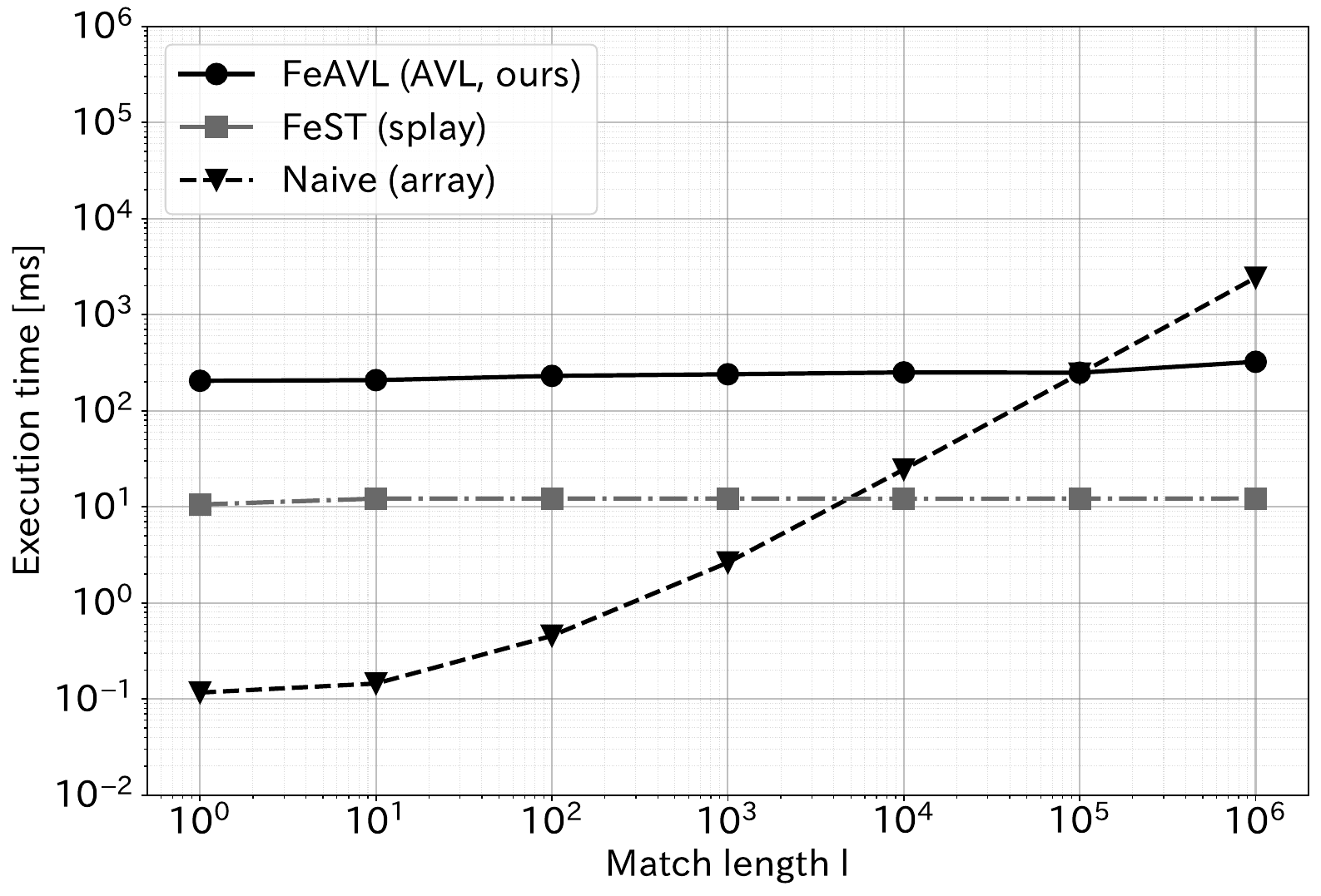}
    \caption{}\label{fig:dyn-equal}
  \end{subfigure}
  \caption{Dynamic operations: (a) extraction and introduction vs.\ the string length $n$; (b) equality vs.\ the match length $\ell$ at $n = 2\times10^6$.}
  \label{fig:dyn}
\end{figure}

\subsection{Persistent Operations}\label{ssec:persistent}

With $Q = 10^4$ queries, persistent introduction scales logarithmically for both trees and linearly for the naive per-version copy; the crossover is near $n \sim 2\times10^4$ (Figure~\ref{fig:per-time}), consistent with Theorem~\ref{thm:persistent}.
For space we issue $Q = K = 10^5$ introductions, retain every resulting version, and vary $n$: per version the trees use $\Oh(\log n)$ and the naive array $\Oh(n)$, so the totals are $\Oh(n_0 + I + U\log n_{\max})$ and $\Oh(K n)$, respectively (Sec.~\ref{sec:persistent}).
At $n = 10^5$ the naive array reaches about $20$\,GB while the trees use about $0.5$\,GB (Figure~\ref{fig:per-space}); the per-node overhead makes the trees heavier only up to $n \sim 10^3$.

\begin{figure}[t]
  \centering
  \begin{subfigure}{0.49\linewidth}
    \centering
    \includegraphics[width=\linewidth]{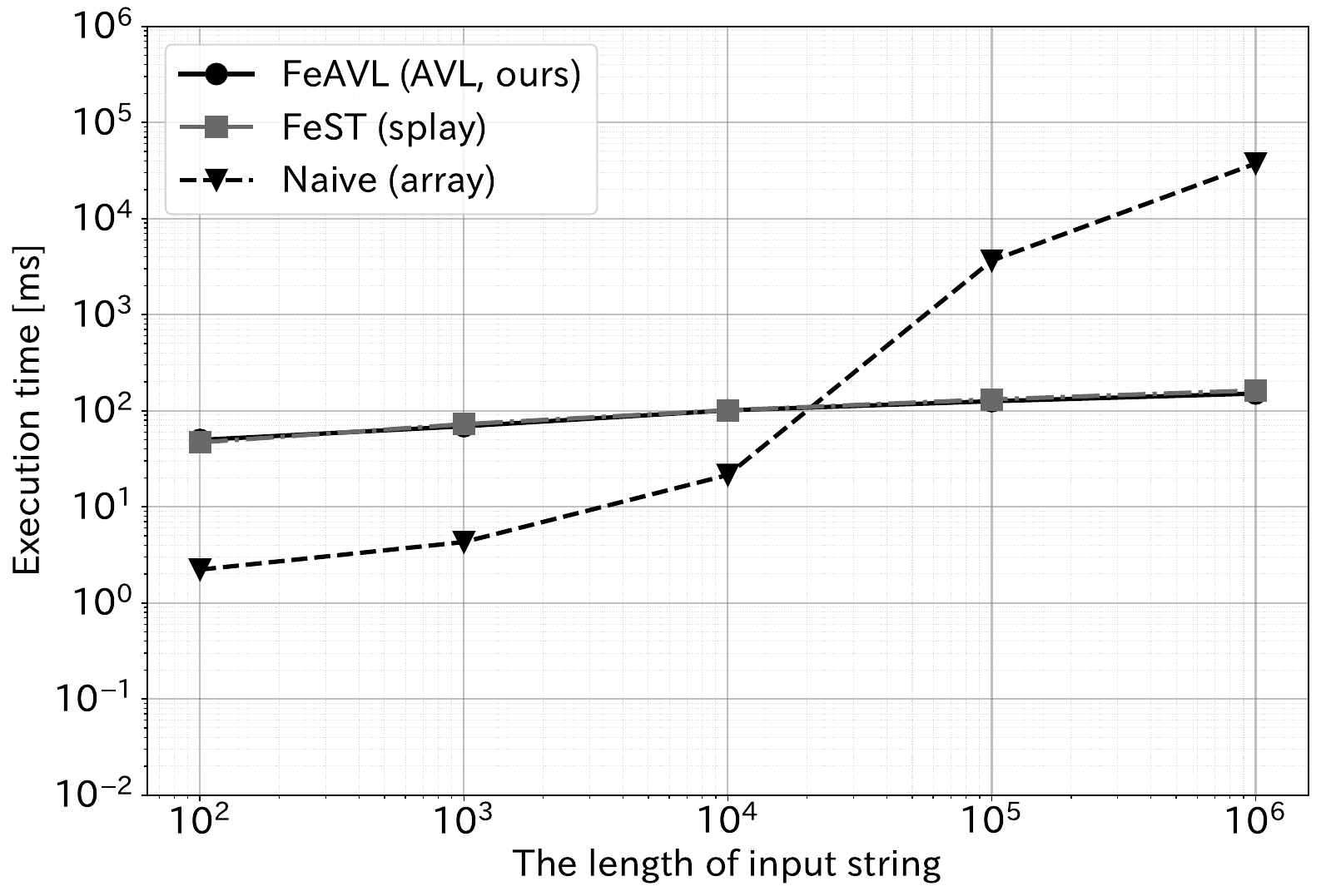}
    \caption{}\label{fig:per-time}
  \end{subfigure}
  \hfill
  \begin{subfigure}{0.49\linewidth}
    \centering
    \includegraphics[width=\linewidth]{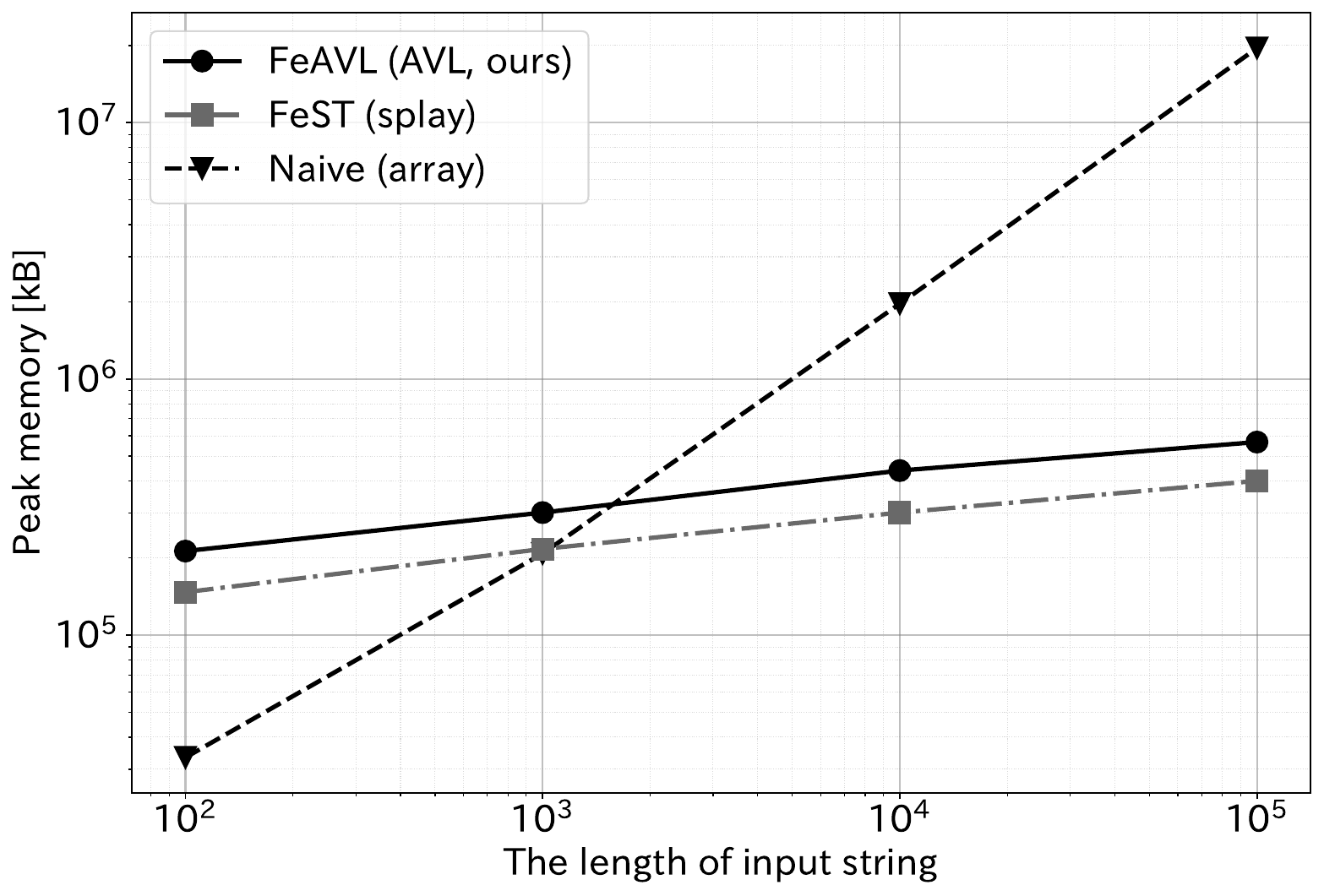}
    \caption{}\label{fig:per-space}
  \end{subfigure}
  \caption{Persistent operations vs.\ the string length $n$: (a) introduction time; (b) peak memory with $K = 10^5$ retained versions.}
  \label{fig:per}
\end{figure}

\subsection{Worst Case of Splay-Based FeST under Persistence}\label{ssec:worst}

To confirm that our \textal{FeAVL} avoids the persistent worst-case degradation of the splay-based \textal{FeST}, we repeatedly apply 
the access pattern 
that triggers it.
We first force the tree into a path by accessing positions in ascending order, 
and then issue $Q = 10^3$ accesses to the deepest position $0$ while varying the size $n$ (Figure~\ref{fig:worst}).
This reflects Proposition~\ref{prop:splaybad}: discarding the splayed result leaves the base version as a path. 
Consequently, the execution time of \textal{FeST} scales linearly ($\sim 42$\,s at $n = 10^5$) 
generating $\Theta(n)$ new nodes, 
whereas \textal{FeAVL} takes near-constant time ($0.2$\,ms) and the naive array takes $\Oh(1)$.



\begin{figure}[t]
  \centering
  \includegraphics[width=0.49\linewidth]{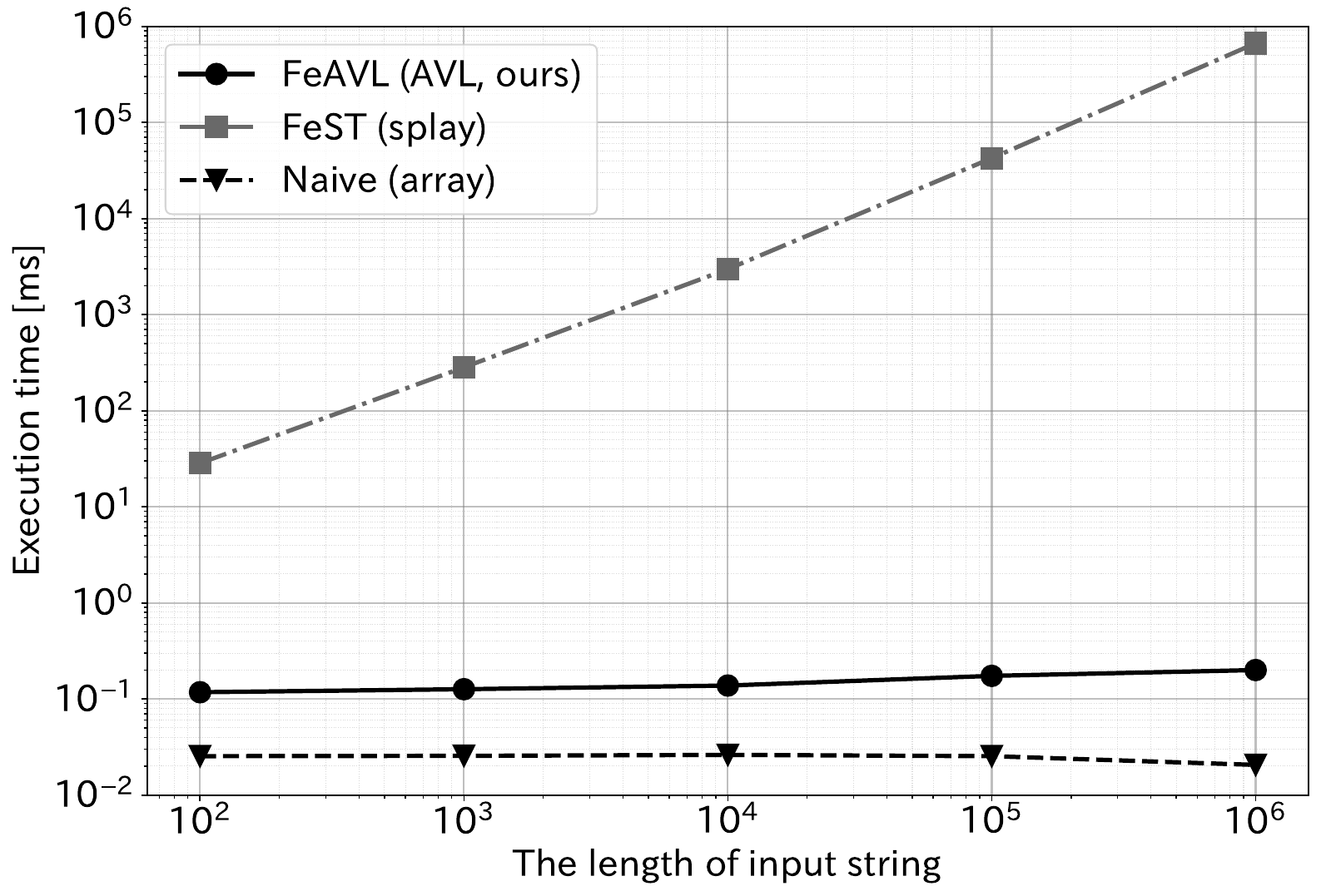}
  \caption{Worst case under persistence: access to the deepest element of a path-shaped version vs.\ the string length $n$.}
  \label{fig:worst}
\end{figure}

\section{Conclusions}\label{sec:conclusion}


We presented \textal{FeAVL} and its grammar-compressed variant to make the FeST paradigm~\cite{liptak2026textbook} fully persistent by replacing splay trees with worst-case balanced structures. This achieves worst-case (near-)logarithmic time for dynamic operations and LCE queries, with the compressed version 
yielding a total size of $\Oh(g_0+I+U \log n_{\max})$, with initial size $g_0=\Oh(\min\{n_0,z_0\log n_0\})$.

Future work includes applying our data structure to large-scale real-world versioned documents like software repositories and genome databases~\cite{makinen2023genome}. A particularly interesting application in bioinformatics is simulating \textit{pangenomic evolution} and \textit{phylogenetic trees}, which models mutations and sequence recombinations (via split and concatenate) while maintaining full historical states~\cite{equi2023algorithms}.





\newpage

\bibliographystyle{splncs04}
\bibliography{ref}

\end{document}